\documentclass[pra,aps,twocolumn,longbibliography]{revtex4-1}
\usepackage[colorlinks=true, citecolor=red, urlcolor=blue ]{hyperref}
\usepackage{graphicx}
\usepackage{bm}
\usepackage{amsmath,amsfonts}
\usepackage{amsthm}
\usepackage{hyperref}
\usepackage{xcolor}
\usepackage{braket}
\theoremstyle{plain}
\usepackage{color}
\usepackage{amssymb}
\usepackage{amsthm}
\usepackage{amsfonts}
\usepackage{float}
\usepackage{tabularx}
\usepackage{graphicx}

\usepackage{mathtools}
\usepackage{esvect}
\usepackage{wrapfig}
\usepackage{amsthm}
\usepackage{verbatim}
\usepackage{bbm}
\usepackage[normalem]{ulem}

\usepackage{enumitem}
\usepackage{fmtcount}
\usepackage{booktabs}
\usepackage{csquotes}
\usepackage{epsfig}

\usepackage{tabularx}
\usepackage{graphicx}
\usepackage[utf8x]{inputenc}
\usepackage{color,soul}
\usepackage{amsmath}
\usepackage{braket}
\usepackage{latexsym}
\usepackage{bm}
\usepackage{graphics,epstopdf}
\usepackage{enumitem}
\usepackage{fmtcount}
\usepackage{booktabs}
\usepackage{csquotes}
\usepackage{epsfig}

\theoremstyle{plain}

\def\bea{\begin{eqnarray}}
\def\eea{\end{eqnarray}}
\def\ba{\begin{array}}
\def\ea{\end{array}}

\def\beq{\begin{equation}}
\def\eeq{\end{equation}}

\usepackage[normalem]{ulem}
\usepackage{float}
\usepackage{graphicx}  
\usepackage{dcolumn}          
\usepackage{amssymb}
\usepackage{appendix}
\usepackage{physics}   
\usepackage{mathtools}
\usepackage{esvect}
\usepackage{wrapfig}
\usepackage{amsthm}
\usepackage{verbatim}
\usepackage{bbm}

\usepackage[mathscr]{euscript}

\def\({\left(}
\def\){\right)}
\def\[{\left[}
\def\]{\right]}

\newcommand{\mc}[1]{\mathcal{#1}}

\newtheorem{theorem}{Theorem}

\newtheorem{definition}{Definition}

\newtheorem{lemma}{Lemma}
\newtheorem{observation}{Observation}

\begin{document}

\title{Preservation of entanglement in local noisy channels}
\author{Priya Ghosh}
\affiliation{Harish-Chandra Research Institute,  A CI of Homi Bhabha National Institute, Chhatnag Road, Jhunsi, Allahabad 211 019, India}

\author{Kornikar Sen}
\affiliation{Harish-Chandra Research Institute,  A CI of Homi Bhabha National Institute, Chhatnag Road, Jhunsi, Allahabad 211 019, India}

\author{Ujjwal Sen}
\affiliation{Harish-Chandra Research Institute,  A CI of Homi Bhabha National Institute, Chhatnag Road, Jhunsi, Allahabad 211 019, India}
\begin{abstract}
Entanglement subject to noise can not be shielded against decaying. But, in case of many noisy channels, the degradation can be partially prevented by using local unitary operations. We consider the effect of local noise on shared quantum states and evaluate the amount of entanglement that can be preserved from deterioration. The amount of saved entanglement not only depends on the strength of the channel but also on the type of the channel, and in particular, it always vanishes for the depolarizing channel. The main motive of this work is to analyze the reason behind this dependency of saved entanglement by inspecting properties of the corresponding channels. In this context, we quantify and explore the biasnesses of channels towards the different states on which they act. We postulate that all biasness measures must vanish for depolarizing channels, and subsequently introduce a few measures of biasness. We also consider the entanglement capacities of channels. We observe that the joint behaviour of the biasness quantifiers and the entanglement capacity explains the nature of saved entanglement. Furthermore, we find a pair of upper bounds on saved entanglement which are noticed to imitate the graphical nature of the latter.

\end{abstract}

\maketitle

\section{Introduction}
\label{sec1}
Entanglement~\cite{horodecki,otfried,sreetama} plays a crucial role in quantum information science. It is used as a resource in many quantum information tasks such as teleportation~\cite{ent-1}, quantum dense coding~\cite{ent-2}, quantum computation~\cite{comp}, entanglement-based quantum cryptography~\cite{crypto,ent-3}, and so on. But it is a very delicate characteristic of shared quantum systems. Realistic systems are subject to noise, and entanglement loss is typically unavoidable therein. Entanglement transformation, and in particular its disappearance, in presence of different classes of noise have been studied by various scientists~\cite{yu1,yu2,ali3,ann,yu3}. Experiments along this line are explored e.g. Refs. ~\cite{kimble,almeida,salles}.

Numerous theoretical~\cite{derkacz,plastina,nemes,petersen,rau3,ali1,ali2,sun,keane,hussain,xiao1,xiao2,liao} as well as experimental~\cite{xu,kim,lim,rau2} researches have been carried out on prevention or reduction of entanglement loss. In recent years, it has been realized that though local unitaries can not change entanglement of any state, its operation can restrain or delay entanglement degradation, even when it is applied only on a single party. If a local unitary is operated afore the system is exposed to noise, the effect of the noise on the system may get reduced compared to the case without the application of any unitary~\cite{rau1,noise3,sinha1,chaves}. This can be exemplified through analyzing multipartite graph states in presence of the local dephasing channel~\cite{noise3}, and bipartite maximally entangled states after transforming through local amplitude damping and local dephasing channels~\cite{chaves}. Since local unitary operations are easy to implement, this method of preservation of entanglement is conceivably an experimentally friendly and low cost process~\cite{laurat,marcelo}. 

Focusing on bipartite states, successful protection of entanglement from local noise acting identically on the two parties, through the help of local unitary operations, depends on the type and strength of the channels.
We name the amount of entanglement that can be saved by applying the optimal local unitary on a single party as ``saved entanglement" (SE). 
We observe that entanglement can not be 
saved by using this method in case of locally covariant channels, e.g. the local depolarizing channel, whereas it is possible to protect a finite amount of entanglement in presence of a local amplitude damping, bit flip, phase flip, or bit-phase flip channel. We numerically obtain that the saved entanglement is exactly same for local bit flip, phase flip, bit-phase flip noise. In this work, we try to explore the reason behind the disparate behavior of saved entanglement for distinct channels.

Since local unitaries just rotate the states locally, if the resulting state is more robust to a noise, the reason must be the noise's partiality towards a set of states having a particular \emph{direction}. This fact motivates us to investigate the property of ``biasness" of channels, which describes the channel's bias towards a bunch of states. We make several observations in this direction, which lead us state a postulate that must be satisfied by a quantifier of biasness of a quantum channel. We subsequently introduce three such quantifiers. We also consider the entanglement capacity of quantum channels. We show that the nature of the
SE can be explained by composing the behaviours of biasness and entanglement capacity. In particular, by considering two-qubit states and some paradigmatic noisy channels, viz. amplitude damping, bit flip, phase flip, and bit-phase flip, acting locally and identically on both of the qubits, we observe that at low noise strengths, SE monotonically increases with biasness of the local channel, whereas it decreases monotonically with entanglement capacity at higher noise strengths. We also present two upper bounds on the saved entanglement. These bounds are seen to mimick the nature of SE.


The rest of the paper is organized as follows. In Sec. \ref{sec2}, we briefly recapitulate definitions of some well-known quantities, which will be needed in the rest of the paper, such as covariant channels, $l_1$-norm, distance between two channels, concurrence, etc. Saved entanglement and entanglement capacity are defined in Sec. \ref{sec3}. To find our way towards defining eligible measures of biasness, we make a few observations about saved entanglement of channels in the same section. Based on these observations, we define various appropriate biasness quantifiers in Sec. \ref{sec4}. We determine two bounds on the saved entanglement in Sec. \ref{sec5}. Different well-known examples of noise are considered in Sec. \ref{sec6}, and their behaviour with respect to the biasness measures as well as entanglement capacity and saved entanglement are obtained and discussed. We present the concluding remarks in Sec. \ref{sec7}.

\section{Prerequisites}
\label{sec2}
In this section we will briefly discuss some basic tools which will be used later.\\ 

Quantum channels transform a state, $\rho$, acting on a Hilbert space, $\mathcal{H}$, to another state, $\rho'$, acting on the same or different Hilbert space, $\mathcal{H}'$. Operation of a quantum channel can be described using a completely positive trace-preserving (CPTP) map, $\Lambda$, and the transformation can be denoted as $\Lambda: \rho \rightarrow \Lambda(\rho)$. Corresponding to every CPTP map there exists a set of Kraus operators, $\{K_i\}_i$,  satisfying $\sum_i K_i^\dagger K_i={I}_d$, such that the transformation $\Lambda(\rho)$ can be expressed as $\Lambda(\rho) = \sum_i K_i \rho K_i^\dagger$~\cite{chuang}. Here, ${I}_d$ is the identity operator on $\mathcal{H}$.

\textbf{Notations:} 
In general, we denote density matrices, unitaries, and noisy channels acting on the Hilbert space, $\mathcal{H}$, of dimension $d$, by $\rho$, $U$, and $\Lambda$ respectively, unless specified otherwise. The set of rank-one states, density matrices, and unitary operators on $\mathcal{H}$ are denoted as $\mathcal{P(H)}$, $\mathcal{S(H)}$, and $\mathcal{U(H)}$ respectively. In case of composite Hilbert spaces, $\mathcal{H}\otimes\mathcal{H}$, of dimension $d\times d$, we use the same notations but in bold symbols, i.e the density matrices, unitaries, and channels are expressed as $\boldsymbol{\rho}$, $\boldsymbol{U}$, and $\boldsymbol{\Lambda}$. 

\begin{definition}
[Covariant channels~\cite{dariusz,winter}]
Suppose that a quantum channel, $\Lambda$, operates on a state $\rho\in\mathcal{S(H)}$ in such a way that $\Lambda(\rho)$ also acts on the same Hilbert space $\mathcal{H}$. Then the channel $\Lambda$ is said to be covariant if the relation,
\begin{equation}
\label{cov-eqn}
  \Lambda\left(U \rho U^\dagger\right) = U \Lambda(\rho) U^\dagger , 
  \end{equation}
is true for all $\rho\in\mathcal{S(H)}$ and $U\in\mathcal{U(H)}$. 
\end{definition}

An example of covariant channels is the 
depolarizing channel~\cite{frey}, $\Lambda_{DC}$, that can be expressed by its action on a state, $\rho\in\mathcal{S(H)}$, 
viz.
\begin{equation}
  \Lambda_{DC} (\rho) = \left(1- p\right) \rho + \frac{p}{d} I_d, \label{eq-dc} 
\end{equation}
where $p$ is the strength of the depolarizing noise. 
Let us prove this statement, for completeness. Any single-qubit state can be represented as a point on or within the Bloch sphere. A unitary operator acting on the state just rotates the directed line (Bloch vector) joining the point from the center of the sphere, whereas the depolarizing channel shrinks the length of the vector keeping its direction fixed. Hence, the action of the unitary and the channel commutes with each other. Thus, a depolarizing channel acting on single-qubit states is a valid example of a covariant channel. 
For a proof that depolarizing channels acting on arbitrary dimensional states are also covariant channels, see Appendix.


\begin{definition}
[$l_1$-norm of matrix] 
The $l_1$-norm of a matrix A having $m$ rows and $n$ columns is defined as
\begin{equation}
\label{l1-norm}
    ||A||_{1} \coloneqq \max_{1 \leq j \leq n} \sum_{i=1}^{m} \abs{a_{ij}},
\end{equation}
where $a_{ij}$ denotes the element of $A$ situated at the intersection of the $i$th row and the $j$th column of $A$.
\end{definition}
Similarly, the $l_1$-norm distance between two matrices of equal order, say $A$ and $B$, can be defined as
\begin{equation}
    ||A-B||_1 \coloneqq \max_{1 \leq j \leq n} \sum_{i=1}^{m} \abs{a_{ij} - b_{ij}}, \label{eq1}
\end{equation}
where  $\{b_{ij}\}_{ij}$ are elements of $B$.

Next, we are going to define a measure of distance between two channels~\cite{dist-map,nielsen,kitaev}. In this regard, let us first recapitulate the Choi–Jamiołkowski–Kraus–Sudarshan (CJKS) isomorphism~\cite{jamio,choi,kraus,sudarshan}.  Consider a ``reference" Hilbert space $\mathcal{H}'$, having the same dimension, $d$, as of $\mathcal{H}$. A maximally entangled state acting on the composite Hilbert-space $\mathcal{H}\otimes\mathcal{H}'$, can be defined as $\ket{\boldsymbol{\phi}^+}=\frac{1}{\sqrt{d}}\sum _{0\leq i<d} \ket{ii}$. Then according to the CJKS isomorphism the map $\Lambda\rightarrow \rho_\Lambda=I_d\otimes \Lambda\left(\ket{\boldsymbol{\phi}^+}\bra{\boldsymbol{\phi}^+}\right)$ is bijective. 

\begin{definition}[Distance between two channels~\cite{dist-map}]
\label{def:dist-channel}
A measure of the distance between two channels, $\Lambda$ and $\Lambda'$, acting on same dimensional states (say $d$), can be defined using the CJKS isomorphism as
\begin{align}
\mathcal{D}(\Lambda \parallel \Lambda') = \mathcal{D}_s\left(\mathbbm{1} \otimes \Lambda \left(\ket{\phi^{+}} \bra{\phi^{+}}\right) \parallel  I_d \otimes \Lambda' \left(\ket{\phi^{+}} \bra{\phi^{+}}\right)\right), \label{dist-eqn}
\end{align}
where $\mathcal{D}_s$ represents any measure of distance between two states. For numerical calculations, we will use the $l_1$-norm as the distance measure, $\mathcal{D}_s$.
\end{definition}

Let us now move to a quantifier of entanglement. Precisely, we consider concurrence, which is an entanglement measure of bipartite quantum states of $\mathbbm{C}^2 \otimes \mathbbm{C}^2$~\cite{concurrence2,concurrence,bennett-1,bennett-2}. The concurrence ($\mathcal{C}$) of any two-qubit density operator, say $\boldsymbol{\rho}_2$, is given by
\begin{equation}
    \mathcal{C}(\boldsymbol{\rho}_2) = \max \lbrace 0, \lambda_1-\lambda_2-\lambda_3-\lambda_4 \rbrace,
\end{equation}
where $\lambda_1,\lambda_2,\lambda_3,\lambda_4$ are the eigenvalues of
$$ \omega = \sqrt{\sqrt{\boldsymbol{\rho}_2} \Tilde{\boldsymbol{\rho}}_2 \sqrt{\boldsymbol{\rho}}_2}$$
in decreasing order, with $\Tilde{\boldsymbol{\rho}}_2 = (\sigma_y \otimes \sigma_y) \boldsymbol{\rho}_2^{*} (\sigma_y \otimes \sigma_y)$.

For all the numerical calculations ahead, we will use concurrence as the measure of entanglement.



\section{Observations}
\label{sec3}
Degradation of entanglement of any composite system is inevitable when acted on by local noisy channels. Despite the incapability of local unitaries to change the entanglement of a system, action of a local unitary on an entangled state can reduce the intensity of the damage unless the system is in a two-qubit pure state and the noise acts on only a single party~\cite{konrad}. 
We first want to quantify the amount of entanglement that can be protected from the claws of noise. In this regard, we consider bipartite states, $\boldsymbol{\rho}$, acting on a composite Hilbert space, $\mathcal{H}\otimes\mathcal{H}$, and use local unitaries of the type $\boldsymbol{U}=I_d \otimes U$, to guard the entangled state from local noise, $\boldsymbol{\Lambda}=\Lambda\otimes\Lambda$. 
 \begin{definition} [Saved entanglement (SE)] 
The maximum amount of entanglement that can be saved by the help of local unitaries of the form $\boldsymbol{U}$ ($=I_d\otimes U$) from a noisy channel $\boldsymbol{\Lambda}$ ($=\Lambda\otimes\Lambda$) with a fixed noise strength can be defined as
\begin{equation}
    \textnormal{SE} \coloneqq  \max_{\boldsymbol{\rho}\in\mathcal{S} (\mathcal{H}\otimes\mathcal{H})} \left[\max_{U\in\mathcal{U(H)}} \mathcal{E}\left(\boldsymbol{\Lambda} \left(\boldsymbol{U} \boldsymbol{\rho} \boldsymbol{U}^\dagger\right)\right) - \mathcal{E}\left(\boldsymbol{\Lambda} (\boldsymbol{\rho})\right)\right] , \label{eq2}
\end{equation}
where $\mathcal{E}$ is any fixed measure of entanglement.
\end{definition}
It is straightforward from the definition that SE is a property of the local noise $\boldsymbol{\Lambda}$, or, more precisely, of the individual single-party noise, $\Lambda$. For a trivial unitary, i.e., for $U=I_d$, the quantity $\mathcal{E}\left(\boldsymbol{\Lambda} \left(\boldsymbol{U} \boldsymbol{\rho} \boldsymbol{U}^\dagger\right)\right) - \mathcal{E}\left(\boldsymbol{\Lambda} (\boldsymbol{\rho})\right)$ is zero for all $\boldsymbol{\rho}\in\mathcal{S(H\otimes H)}$. Since SE is defined as the maximum value of this quantity, maximized over all unitaries (and states), it is obvious that SE will always be greater or equal to zero.

Restricting ourselves to the composite Hilbert space of dimension $2\otimes2$ and numerically optimizing over the set of pure states only, we obtain the following results.

\begin{observation}
SE is zero for the local depolarizing channel, $\Lambda_{DC}\otimes\Lambda_{DC}$, for all considered values of the noise strength, $p$.
\end{observation}


\begin{observation}
Local amplitude damping channel~\cite{chaves} 
has non-zero SE for all noise strengths, $p$, except  $p = 0$ and $p = 1$. 
\end{observation}


\begin{observation}
Behaviour of SEs of local bit flip, phase flip (dephasing)~\cite{chaves}, and bit-phase flip channels are numerically found to be identical, and is again non-zero for non-zero and non-unit noise strengths, $p$. 
\end{observation}
In Fig. \ref{fig1}, the behaviour of \textnormal{SEs} for different local channels, viz. amplitude damping\textcolor{blue}{,} bit flip, phase flip, and bit-phase flip channels, are exhibited as functions of the corresponding noise strengths, $p$. It is clear from the figure that SEs of bit flip, phase flip, and bit-phase flip channels are equal at every noise strengths. All though the saved entanglements are quantitatively distinct for different channels, they have qualitative similarities. In particular, they monotonically increase with noise strength up to a cut off value after which they start to decrease with $p$. The value of the SE at the cut-off, depends on the form of the noise.

\begin{figure}[h!]
\includegraphics[scale=1.0]{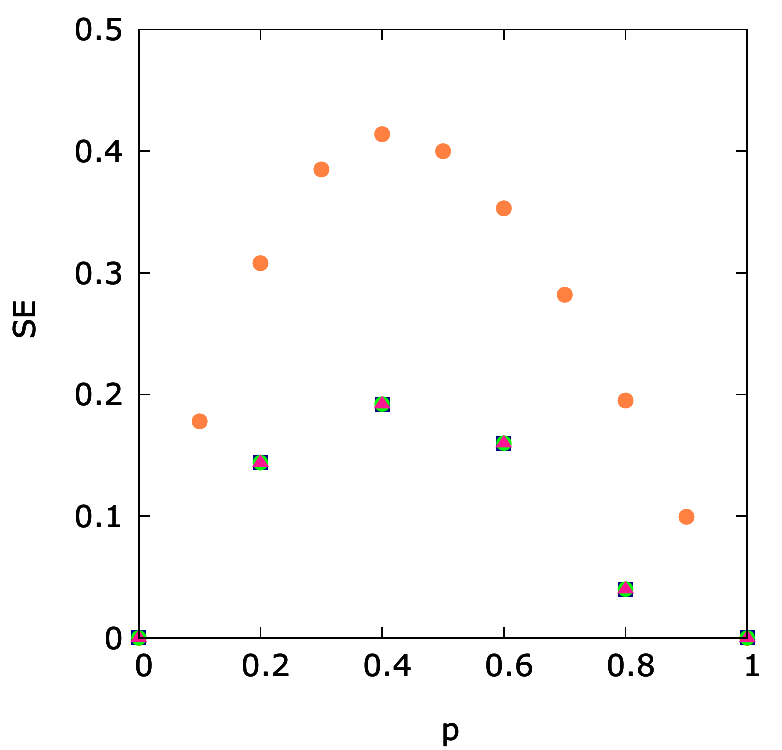}
\caption{Saved entanglements for paradigmatic channels. We exhibit the amounts of entanglement that can be protected i.e., SEs, on the vertical axis, as functions of noise strength, $p$ (horizontal axis). The yellow, blue, green, and pink coloured points represent the values of SE for amplitude damping, bit flip, phase flip, and bit-phase flip channels respectively. The horizontal axis is dimensionless whereas the vertical axis is in ebits.} 
\label{fig1} 
\end{figure}

\begin{definition}
[Entanglement capacity ($\textnormal{EC}$)] 
We define the maximum amount of entanglement that survives after the application of a local noisy channel, $\boldsymbol{\Lambda}$, on an arbitrary state, $\boldsymbol{\rho}$, i.e., 
\begin{equation}
\label{ent-cap-eqn}
    \textnormal{EC} \coloneqq \max_{\boldsymbol{\rho} \in \mathcal{S(H\otimes H)}} \mathcal{E}(\boldsymbol{\Lambda}(\boldsymbol{\rho})),
\end{equation}
as entanglement capacity of the channel, $\boldsymbol{\Lambda}$.
\end{definition}

The value of \textnormal{EC} depends on the strength, $p$, of the channel. It is usually a decreasing function of $p$ because as $p$ increases, the channel becomes more noisy and thus destroys the entanglement more. In case of amplitude damping, bit flip, phase flip, and bit-phase flip channels, EC=0 for $p=1$.

\section{Measures of biasness}
\label{sec4}

When a channel affects individual states differently, we say the channel is biased. In this section, we will discuss various methods of quantifying the biasness. Before going into the details of the measures, let us first state an intuitively satisfactory postulate for a function to be an acceptable measure of biasness. \\

\noindent \textbf{Postulate:}
The value of the measure of biasness should be zero for the depolarizing channel.\\


\noindent We note that the identity channel is a particular case of a depolarizing channel. In the following subsections, we construct three kinds of biasness measures of channels $\Lambda$ or $\boldsymbol{\Lambda}$, all of which satisfy the above mentioned property. 
Let us mention here that while in this paper we will be concerned with exclusively utilizing the biasness measures for characterizing and understanding saved entanglement, we believe that the concept of biasness and the measures thereof will have a wider applicability.

\subsection{Distance from depolarizing channel}
Any single-qubit system's state can be represented by a point on or inside the Bloch sphere. We know that the action of the depolarizing channel, $\Lambda_{DC}$, on the state will contract the length of the distance between the point and the center of the sphere. The amount of contraction depends on the initial length only and not on the direction of the point. Thus we can say that the channel does not have any biasness towards the direction of the point. This fact motivates us to introduce the following measure of biasness: the distance of a channel from the depolarizing channel (DDC). To evaluate the distance between two channels we use the measure defined in the previous section, i.e., Eq. \eqref{dist-eqn}. 
Hence, the biasness, DDC, of a channel, $\Lambda$, can be mathematically expressed as
\begin{equation}
    \textnormal{DDC}(\Lambda) \coloneqq \max_{p} \mathcal{D}(\Lambda_{DC} || \Lambda), \label{eq3}
\end{equation}
where $p$ denotes strength of the depolarizing channel.

Though the measure is introduced based on single-qubit depolarizing channels, it can be generalized to higher dimensions. DDC$(\Lambda)$, from the definition itself, is zero for a depolarizing channel. For the strength $p=0$, the depolarizing channel becomes equivalent to the identity channel [see Eq. \eqref{eq-dc}], and thus DDC is also zero for the identity channel.

In the next section, we will show that the biasness measure, DDC, of different channels, precisely, amplitude damping, bit flip, phase flip, and bit-phase flip, is correlated with the amount of entanglement saved with the help of local unitaries.


\subsection{Channel's dependence on state}
To examine how the transformation of a state, by a channel, depends on the direction of the input state, we can define a channel's dependence on state (CDS). Let $\rho$ and $\rho^\perp$ be two orthogonal pure qubit states. Then, the CDS of a channel, $\Lambda$, is defined as
\begin{equation}
\label{max-min-eqn}
    \textnormal{CDS}(\Lambda) \coloneqq \max_{\rho\in \mathcal{P(H)}} \left[F\left(\rho,\rho^\perp\right)\right] - \min_{\rho\in \mathcal{P(H)}}\left[F\left(\rho,\rho^\perp\right)\right], 
\end{equation}
where $F\left(\rho,\rho^\perp\right)=\text{Tr}\left(\Lambda\left(\rho\right)\rho\right) + \text{Tr}\left(\Lambda\left(\rho^\perp\right)\rho^\perp\right)$. Since for the identity channel, $\Lambda(\rho)\rightarrow\rho$, it is straightforward that for the identity channel, $\textnormal{CDS} = 0$. For the depolarizing channel, $F(\rho,\rho^\perp)=2\left[1-p\left(\frac{d-1}{d}\right)\right]$, which is independent of $\rho$, for a given dimension, which implies CDS$(\Lambda_{DC})$ = 0 for all $p$. Hence, CDS can be a contender for measuring biasness.

\subsection{Incovariance}
Before going into the discussion about the next quantifier of biasness, let us first state a theorem. 
\begin{theorem}
\label{cov-theo}
Saved entanglement is always zero for local covariant channels.
\end{theorem}

\begin{proof}
 Consider a bipartite system, and the corresponding composite Hilbert space, $\mathcal{H}\otimes\mathcal{H}$, of dimension $d\times d$. Saved entanglement of a local channel, $\boldsymbol{\Lambda}=\Lambda\otimes\Lambda$, is
\begin{equation*}
\textnormal{SE} = \max_{\boldsymbol{\rho} \in S(\mathcal{H}\otimes\mathcal{H})} \left[ \max_{U\in\mathcal{U(H)}} \mathcal{E}\left(\boldsymbol{\Lambda} \left(\boldsymbol{U} \boldsymbol{\rho} \boldsymbol{U}^\dagger\right)\right) - \mathcal{E}\left(\boldsymbol{\Lambda} \left(\boldsymbol{\rho}\right)\right) \right],
\end{equation*}
where $\boldsymbol{U}=I_d\otimes U$ is the local unitary used to save the entanglement. Let us assume that $\Lambda$ is covariant, i.e $\Lambda\left(U\rho U^\dagger\right)=U\Lambda(\rho)U^\dagger$, for all $\rho\in\mathcal{S(H)}$. Then we can replace $\Lambda\otimes\Lambda\left[\left(I_d\otimes U\right)\boldsymbol{\rho}\left( I_d\otimes U^\dagger\right)\right]$ by $I_d\otimes U[\Lambda\otimes\Lambda(\boldsymbol{\rho})] I_d\otimes U^\dagger$. Thus we have
\begin{eqnarray*}
\textnormal{SE}
&=& \max_{\boldsymbol{\rho} \in S(\mathcal{H}\otimes\mathcal{H})} \left[ \max_{U\in\mathcal{U(H)}} \mathcal{E}\left(\boldsymbol{U} \boldsymbol{\Lambda}(\boldsymbol{\rho}) \boldsymbol{U}^\dagger\right) - \mathcal{E}(\boldsymbol{\Lambda} ( \boldsymbol{\rho})) \right] ,\\
&=& \max_{\boldsymbol{\rho} \in \mathcal{S(H\otimes H)}} \left[ \mathcal{E}(\boldsymbol{\Lambda} (\boldsymbol{\rho})) - \mathcal{E}(\boldsymbol{\Lambda} (\boldsymbol{\rho})) \right] ,\\
&=& 0.
\end{eqnarray*}
Here we have used the fact that entanglement remains unchanged under local unitary operations.
\end{proof}

Since local covariant channels can never display a non-zero saved entanglement for any noise strength, local incovariance of channels (IC) can be a reason of exhibition of non-zero saved entanglement, and therefore can be another quantifier of biasness. The mathematical definition of IC can be 
\begin{equation}
    \textnormal{IC}(\boldsymbol{\Lambda}) \coloneqq \max_{{U} \in \mc{{U(H)}}, \boldsymbol{\rho} \in \mathcal{S(H\otimes H)}} ||\boldsymbol{\Lambda}( \boldsymbol{U} \rho  \boldsymbol{U}^\dagger) -   \boldsymbol{U} \boldsymbol{\Lambda} (\rho) \boldsymbol{U}^\dagger||_1, \label{eq5}
\end{equation}
 where $\boldsymbol{U}=I_d\otimes U$. Both identity and any other depolarizing channels are covariant, and thus IC satisfies the desirable postulate for being a measure of biasness. We will analyze IC for different channels in the succeeding section, and in those calculations, we will optimize over the set of pure states only instead of considering the whole set, $\mathcal{S(H\otimes H)}$.

\section{Bounds on saved entanglement}
\label{sec5}
In this part, we will introduce two bounds on the saved entanglement. Let us consider a bipartite state, $\boldsymbol{\rho}$, and let the noise acting on the state be $\boldsymbol{\Lambda}$. Moreover, let us suppose that only the second party applies the unitary operator to protect entanglement. Therefore, the form of the applied local unitary is $\boldsymbol{U}\equiv I_d\otimes U$. Let $U_{\textnormal{max}}$ and $\boldsymbol{\rho}_{\textnormal{\textnormal{max}}}$ be the unitary operator and the bipartite pure state respectively for which the optimization in Eq. \eqref{eq2} can be achieved. Then the saved entanglement of the channel $\boldsymbol{\Lambda}$ is
\begin{equation}
\label{es-eqn}
    \textnormal{SE} \coloneqq    \mathcal{E}\left(\boldsymbol{\Lambda} \left(\boldsymbol{U}_{\textnormal{max}} \boldsymbol{\rho}_{\textnormal{max}} \boldsymbol{U}_{\textnormal{max}}^\dagger\right)\right) - \mathcal{E}(\boldsymbol{\Lambda} (\boldsymbol{\rho}_{\textnormal{max}})),
\end{equation}
where $\boldsymbol{U}_{\textnormal{max}}=I_d\otimes U_{\textnormal{max}}$. As discussed in Sec. \ref{sec3}, the above quantity is greater or equal to zero. Thus we have
\begin{equation}
\label{ana-e1}
    \mathcal{E} \left(\boldsymbol{\Lambda} \left(\boldsymbol{U}_{\textnormal{max}} \boldsymbol{\rho}_{\textnormal{max}} \boldsymbol{U}^{\dagger}_{max}\right)\right) \geq \mathcal{E} (\boldsymbol{\Lambda} (\rho_{\textnormal{max}})).
\end{equation}
The operator $\boldsymbol{\Lambda} (\boldsymbol{U}_{\textnormal{max}} \boldsymbol{\rho}_{\textnormal{max}} \boldsymbol{U}^{\dagger}_{max})$ is a density matrix, and thus can be decomposed in terms of two density matrices in the following way:
\begin{equation}
\label{ana-e2}
   \boldsymbol{\Lambda} \left(\boldsymbol{U}_{\textnormal{max}} \boldsymbol{\rho}_{\textnormal{max}} \boldsymbol{U}^{\dagger}_{max}\right) = p_1 (\boldsymbol{\Lambda}( \boldsymbol{\rho}_{\textnormal{max}})) + p_2 \boldsymbol{\rho}',
\end{equation}
where $p_1$, $p_2\geq 0$ and $p_1 + p_2 =1$. A trivial solution of the above equation is $p_1=0$, $p_2=1$, and $\boldsymbol{\rho}'=\boldsymbol{\Lambda} \left(\boldsymbol{U}_{\textnormal{max}} \boldsymbol{\rho}_{\textnormal{max}} \boldsymbol{U}^{\dagger}_{max}\right)$. But there can be multiple solutions.

Let us now restrict ourselves to the entanglement quantifiers which satisfy the convexity property ~\cite{virmani}. Concurrence~\cite{concurrence,concurrence2}, relative entropy of entanglement~\cite{knight,plenio,relative2}, negativity~\cite{negativity} are some examples of such quantifiers. Using the convexity property and the expression given in Eq. \eqref{ana-e2}, we can write 
\begin{align}
     \mathcal{E}(\boldsymbol{\Lambda}(\boldsymbol{U}_{\textnormal{max}} \boldsymbol{\rho}_{\textnormal{max}} \boldsymbol{U}^{\dagger}_{\textnormal{max}})) = \mathcal{E} (p_1 (\boldsymbol{\Lambda}( \boldsymbol{\rho}_{\textnormal{max}})) + p_2 \boldsymbol{\rho}') \nonumber\\
     \leq p_1 \mathcal{E}(\boldsymbol{\Lambda} (\boldsymbol{\rho}_{\textnormal{max}})) + p_2 \mathcal{E}(\boldsymbol{\rho}'). \label{eq6}
\end{align}
Then an upper bound on the SE of channels can be determined as
\begin{align}
\textnormal{SE} &=  \mathcal{E}\left(\boldsymbol{\Lambda} \left(\boldsymbol{U}_{\textnormal{max}} \boldsymbol{\rho}_{\textnormal{max}} \boldsymbol{U}^\dagger_{\textnormal{max}}\right)\right) - \mathcal{E}(\boldsymbol{\Lambda} (\boldsymbol{\rho}_{\textnormal{max}})) \nonumber\\
& \leq  p_1 \mathcal{E}(\boldsymbol{\Lambda} (\boldsymbol{\rho}_{\textnormal{max}})) + p_2 \mathcal{E}(\boldsymbol{\rho}') - \mathcal{E}(\boldsymbol{\Lambda} (\boldsymbol{\rho}_{\textnormal{max}}))\nonumber\\
&\leq p_2 \mathcal{E}(\boldsymbol{\rho}') - (1-p_1) \mathcal{E}(\boldsymbol{\Lambda} (\boldsymbol{\rho}_{\textnormal{max}})) \nonumber\\ 
&\leq p_2 \lbrack \mathcal{E}(\boldsymbol{\rho}') -  \mathcal{E}(\boldsymbol{\Lambda} (\boldsymbol{\rho}_{\textnormal{max}})) \rbrack \nonumber\\
&\leq \min \left[p_2 \left[ \mathcal{E}(\boldsymbol{\rho}') -  \mathcal{E}(\boldsymbol{\Lambda} (\boldsymbol{\rho}_{\textnormal{max}})) \right]\right] \label{ana-e3}\\  
& \leq \min \left[p_2 \mathcal{E}(\boldsymbol{\rho}')\right].\label{ana-e4}
\end{align}
The minimization is over all possible decompositions expressed in Eq. \eqref{ana-e2}. From inequalities \eqref{ana-e1} and \eqref{eq6}, we see that $\boldsymbol{\rho}'$ can not be separable.

There can be numerous pairs of $\{\boldsymbol{U}_{\textnormal{max}},\boldsymbol{\rho}_{\textnormal{max}}\}$ for which the optimization introduced in the definition of SE is achievable. All of the pairs $\{\boldsymbol{U}_{\textnormal{max}},\boldsymbol{\rho}_{\textnormal{max}}\}$ will satisfy both the inequalities \eqref{ana-e3} and \eqref{ana-e4}. Thus to get the tightest bound we have to minimize the right hand side of those inequalities over the set of pairs $\{\boldsymbol{U}_{\textnormal{max}},\boldsymbol{\rho}_{\textnormal{max}}\}$. Thus we define the following two quantities,
\begin{align}
    \textnormal{EB1} &\coloneqq \min_{\boldsymbol{U}_{\textnormal{max}}, \boldsymbol{\rho}_{\textnormal{max}},p_2} \left[ p_2 \mathcal{E}(\boldsymbol{\rho}') - p_2 \mathcal{E}(\boldsymbol{\Lambda} (\boldsymbol{\rho}_{max})) \right] , \label{ana-fi-1}\\
    \textnormal{EB2} &\coloneqq \min_{\boldsymbol{U}_{\textnormal{max}}, \boldsymbol{\rho}_{\textnormal{max}},p_2}   p_2 \mathcal{E}(\boldsymbol{\rho}'),\label{ana-fi-2} 
\end{align}
which describe bounds on SE.

In case of the identity and other depolarizing channels, the amount of saved entanglement is vanishing. Hence, in those cases, we can choose $\boldsymbol{U}_{\textnormal{max}}$ to be the identity matrix. Therefore, the solution of Eq. \eqref{ana-e2} for which the bounds given in inequality \eqref{ana-e3} and \eqref{ana-e4} are optimal is $\{p_1,p_2\} = \{1,0\}$. Thus we see that EB1 and EB2 are zero for depolarizing channels of all noise strengths.

\section{Biasness and entanglement capacity as an escort to saved entanglement}
\label{sec6}

In this section, we will discuss how biasness and entanglement capacity are correlated with the behaviour of entanglement saved against certain channels, $viz.$ amplitude damping, bit flip, phase flip, and bit-phase flip. To have an overall idea about their relation, in Fig. \ref{fig2}, we present a schematic diagram of the behaviour of the functions. The saved entanglement, for all the considered noisy channels, shows a parabolic nature. It can be seen that the value of SE initially increases with noise strength up to a certain cut-off value. This can be explained through the nature of biasness of the corresponding channel which also is a monotonically increasing function of the same. Since biasness demonstrates the dependence of the channel on the initial state, it indicates that appropriately changing the initial state will alter the effect of the noise, resulting in less entanglement degradation. Thus, more the biasness, more is the possibility of securing entanglement. But after reaching the cut-off value, SE starts decreasing with noise strength. The reason behind this deterioration can be the effect of the intense noise on the initial states, which in turn immensely affects the entanglement of the states, making the states almost separable. That is, though the channel's impact on the states depends on the states themselves, but the outputs have one thing in common: poor entanglement. Thus the amount of saved entanglement, for smaller values of noise strength, follows the behaviour of biasness, whereas for higher values of noise strength, it follows the nature of entanglement capacity. To grasp the characteristics in more detail, we discuss some exemplar noise models in the following sub-sections.

\begin{figure}[h!]
\includegraphics[scale=0.7]{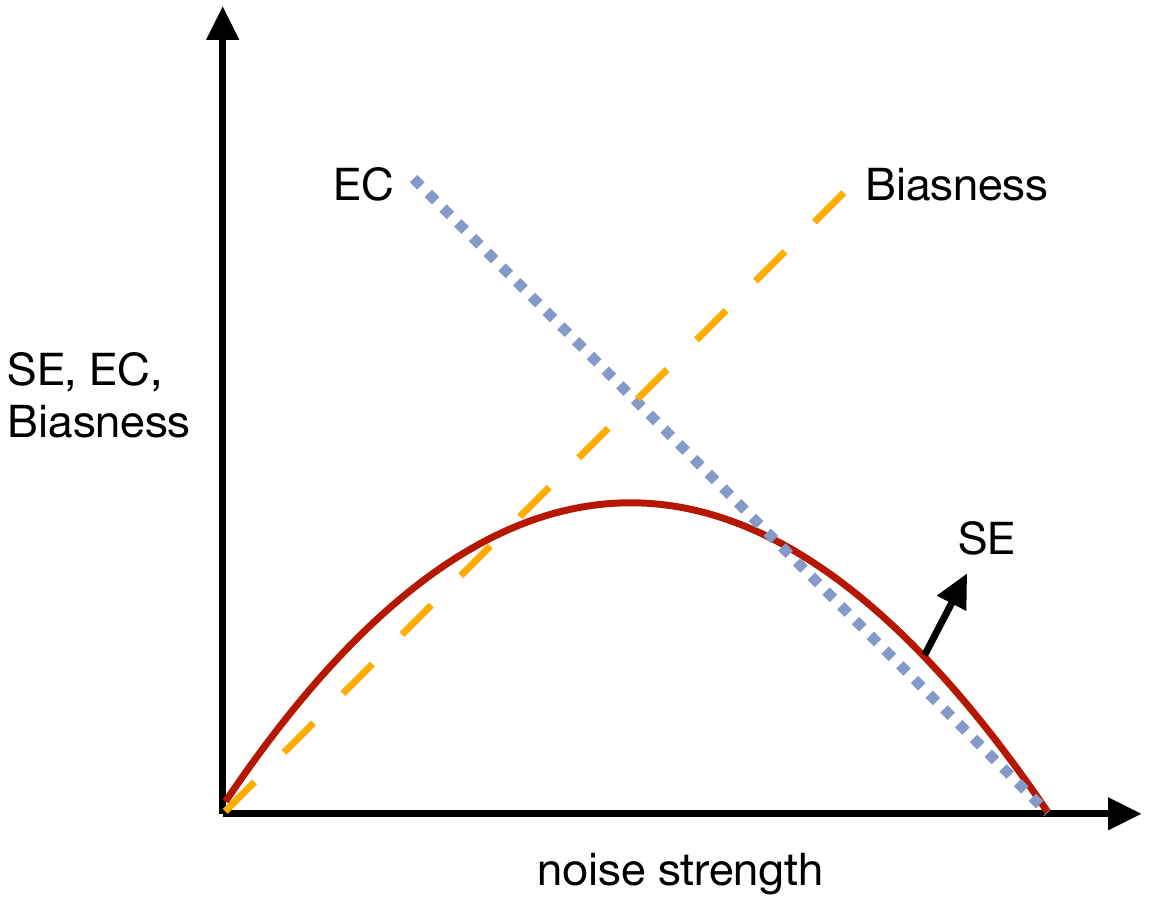}
\caption{Graphical nature of entanglement capacity, biasness, and saved entanglement. A schematic diagram is drawn to represent the behavior of the three distinct functions, viz. EC, SE, and a quantifier of biasness (vertical axis) with respect to noise strength (horizontal axis) of applied local noise. The robustness of entanglement against applied noise decreases with corresponding noise strength, whereas biasness of that noise towards individual input states increases with the strength. The amount of saved entanglement follows the nature of biasness at lower noise strengths and entanglement capacity at higher noise strengths. The quantities are qualitatively represented using yellow dashed line (biasness), blue dotted line (EC), and red solid line (SE). SE and EC are in ebits, while other quantities are dimensionless.} 
\label{fig2} 
\end{figure}

In the following sub-sections, we consider two-qubit systems and apply local noise of the form $\boldsymbol{\Lambda}=\Lambda\otimes\Lambda$, where $\Lambda$ represents a typical noisy channel, for example, amplitude damping channel, bit flip channel, etc. To protect the entanglement, we consider local unitaries of the form $I_2\otimes U$, where $U$ is a single qubit unitary. To determine SE, we optimize over the set of pure states, $\mathcal{P(H\otimes H)}$.

\subsection{Amplitude damping channel}
Let us first consider the amplitude damping channel, $\Lambda_{AD}$. The Kraus operators of the channel is given by
$$K_0 = \left( \begin{array}{cc}
1 & 0 \\
0 & \sqrt{1-p} \end{array} \right), K_1 = \left( \begin{array}{cc}
0 & \sqrt{p} \\
0 & 0 \end{array} \right).$$
Thus the corresponding map can be described as
\begin{equation*}
    \rho\rightarrow \Lambda_{AD}(\rho)=\sum_{i=0}^1 K_i \rho K_i^\dagger.
\end{equation*}
The amount of entanglement that can be saved using local unitaries, i.e., SE, and the biasness quantifiers of the channel, viz. DDC, CDS, and IC, are determined using Eqs. \eqref{eq2}, \eqref{eq3}, \eqref{max-min-eqn}, and \eqref{eq5} respectively. We have used a numerical non-linear optimizer to optimize the functionals. In Fig. \ref{fig3}, we plot SE using brown star points. It is the same curve that was plotted in Fig. \ref{fig1} using yellow circular points. It can be seen that the value of SE increases with $p$ for smaller values of $p$ and then for $p>0.5$ it starts decreasing. We have also plotted the measures of biasness, i.e., DDC, CDS, IC, and entanglement capacity, EC, in the same figure, i.e., Fig. \ref{fig3}. We see that all biasness measures increase with $p$ whereas the entanglement capacity decreases. It is clearly visible that for lower values of $p$, the nature of SE follows the behaviour of biasness and for higher values of $p$, it follows EC. We also plot the bounds EB1 and EB2, expressed in Eqs. \eqref{ana-fi-1} and \eqref{ana-fi-2}, in the same figure. Because of computational limitations, to calculate the bounds, we have not minimized over all $\boldsymbol{U}_{\textnormal{max}}$ and $\boldsymbol{\rho}_{\textnormal{max}}$ but have found only three different pairs of $\{\boldsymbol{U}_{\textnormal{max}}$, $\boldsymbol{\rho}_{\textnormal{max}}\}$ and determined the corresponding EB1 and EB2. From the figure, it is evident that $\textnormal{EB1}$ or $\textnormal{EB2}$ alone can reflect the behaviour of the saved entanglement for all noise strengths.

Interestingly, numerically we have got the same values of the right hand sides of inequalities \eqref{ana-e3} and \eqref{ana-e4}, for each noise strength. Thus we can conclude that $\mathcal{C}(\Lambda (\boldsymbol{\rho}_{\textnormal{max}}))$ is zero for all noise strengths of the channel.


\begin{figure}[h!]
\includegraphics[scale=1.0]{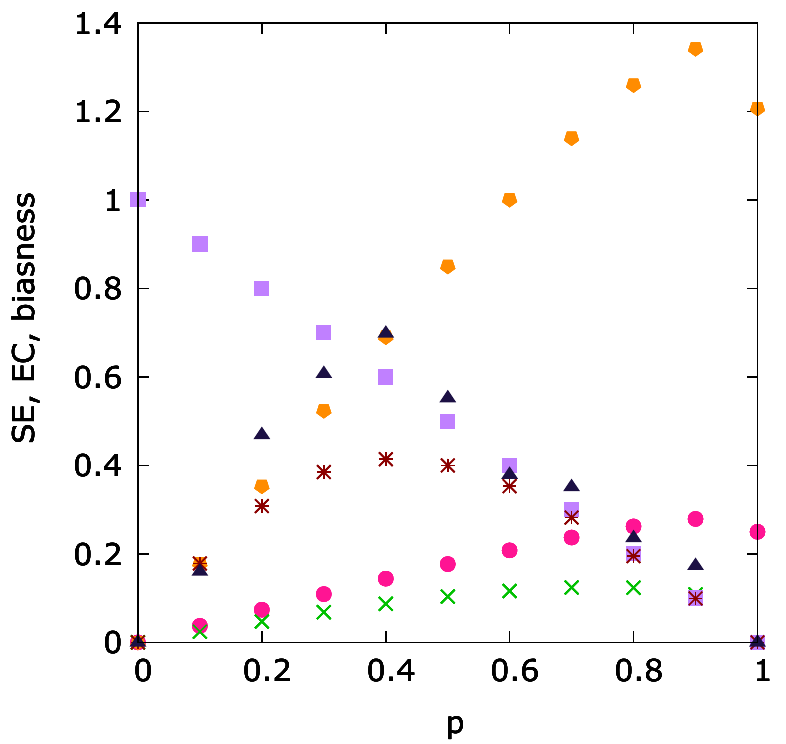}
\caption{Effect of local amplitude damping noise on bipartite pure states and biasness of the channel. We plot SE (brown stars), EC (violet squares), and the biasness measures, i.e., DDC (pink circles), CDS (green crosses), and IC (yellow pentagons) on the vertical axis against the corresponding noise strength, $p$, of applied local amplitude damping channel, represented on the horizontal axis. The bound on saved entanglement, EB1 (black triangles) which is numerically equal to EB2, is also plotted on the same vertical axis. SE, EC, and EB1 are plotted using the dimension of ebits whereas other quantities are dimensionless.}
\label{fig3} 
\end{figure}

\subsection{Bit flip channel}
Next we move to the bit flip noise, $\Lambda_{BF}$, in presence of which, the eigenstates of the $\sigma_z$ matrix, that are $\ket{0}$ and $\ket{1}$, get exchanged with each other, with a finite probability, $\frac{p}{2}$. This transformation can be mathematically expressed as 
$$\rho\rightarrow\Lambda_{BF}(\rho)=\sum_{i=0}^1 K_i \rho K_i^\dagger,$$
where
$$K_0 = \sqrt{1- \frac{p}{2}} \left( \begin{array}{cc}
1 & 0 \\
0 & 1 \end{array} \right), K_1 =  \sqrt{\frac{p}{2}}  \left( \begin{array}{cc}
0 & 1 \\
1 & 0 \end{array} \right).$$

Fig. \ref{fig4} portrays the behavior of the same functionals, as in Fig. \ref{fig3}, for the amplitude damping channel, viz. SE, DDC, CDS, IC, and EC for the bit flip channel, against the noise strength, $p$. It is apparent from the figure that the quantifiers of biasness (DDC, CDS, and IC) and the amount of saved entanglement (SE) behave analogously within the range $0 \leq p \leq \frac{1}{2}$, that is, all of them increase with $p$. After $p = \frac{1}{2}$, the values of the biasnesses continue to increase whereas the corresponding value of SE starts to decrease monotonically. Thus in this range, $\frac{1}{2} \leq p\leq 1$, the nature of SE and EC are alike. Hence we can argue that at first, SE increases because of the presence of biasness in the channel at low noise strength, and then its value starts to reduce at higher values of $p$, because of corresponding low EC of the channel.


\begin{figure}[h!]
\includegraphics[scale=1.0]{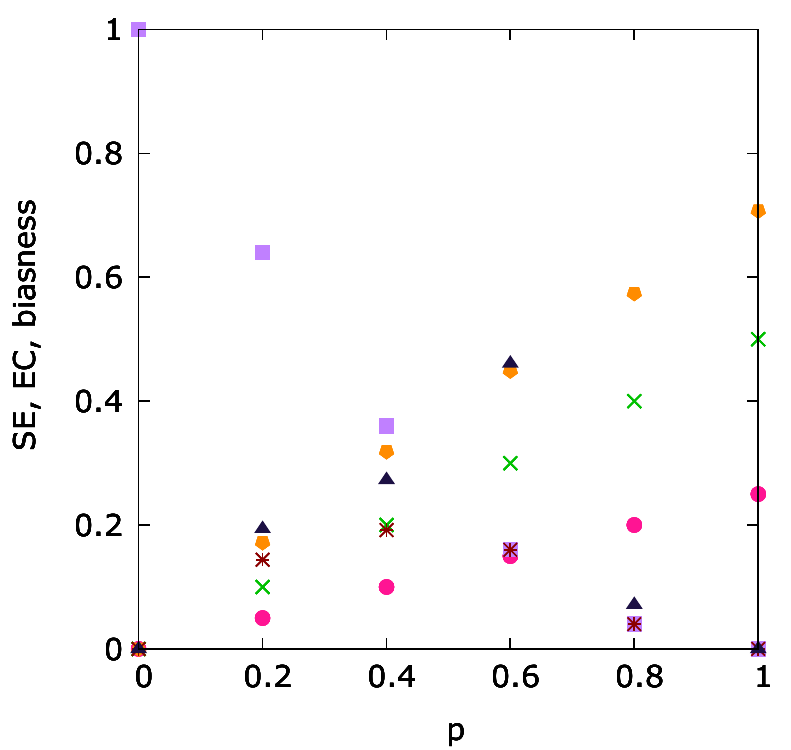}
\caption{Behavior of saved entanglement, entanglement capacity, and measured biasnesses for the bit flip channel. All considerations are the same as in Fig. \ref{fig3} except the fact that here the noise under discussion is bit-flip.} 
\label{fig4} 
\end{figure}

We have calculated EB1 and EB2 for the bit flip channel, only for one pair of $\{\boldsymbol{U}_{\textnormal{max}}$, $\boldsymbol{\rho}_{\textnormal{max}}\}$ and have not minimized over all $\boldsymbol{U}_{\textnormal{max}}$, $\boldsymbol{\rho}_{\textnormal{max}}$. We see that EB1 of the bit flip channel again coincides with EB2 of the same channel for all considered values of the noise strength, $p$. We draw EB1 or EB2 for different values of $p$ in Fig. \ref{fig4}. We see that the bound, EB1 or EB2, can describe the behaviour of SE for all noise strengths.

\subsection{Phase flip channel}
The phase flip noise probabilistically changes the phase of the $\sigma_z$ basis, that is $\ket{0}\rightarrow \ket{0}$ and $\ket{1}\rightarrow -\ket{1}$. The Kraus operators corresponding to the phase flip noise are $K_0=\sqrt{1-\frac{p}{2}} I_2$, $K_1=\sqrt{\frac{p}{2}}\sigma_z$. Thus the transformation can be written as
\begin{equation*}
  \rho\rightarrow\Lambda_{PF} (\rho) = \left(1-\frac{p}{2}\right) \rho + \frac{p}{2}  \sigma_z \rho \sigma_z.
\end{equation*}

It is clear from Fig. \ref{fig1} that SE of the phase flip channel is identical to that of the bit flip channel. We plot the SE again in Fig. \ref{fig5}, to compare it with the biasness measures, DDC, CDS, IC, and the entanglement capacity, EC, which have also been plotted in the same figure, with respect to $p$. We can clearly see from Fig. \ref{fig5} that all the biasness quantifiers discussed in Sec. \ref{sec4} behave similarly as SE of the phase flip channel in lower noisy regions ($p \leq \frac{1}{2}$). But at higher values of the noise strength ($\frac{1}{2} \leq p \leq 1.0$), the saved entanglement seemingly starts to get controlled by entanglement capacity. 

\begin{figure}[h!]
\includegraphics[scale=1.0]{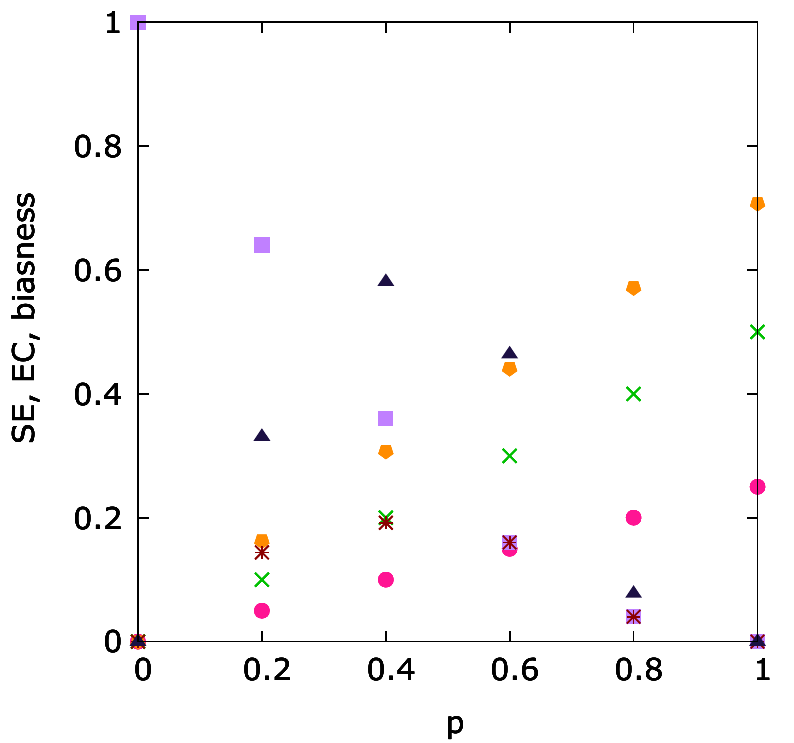}
\caption{Nature of saved entanglement, entanglement capacity, and biasness of phase flip noise. We plot noise strength, $p$, of the phase flip channel along the horizontal axis while the corresponding saved entanglement (SE), entanglement capacity (EC), all biasnesses (DDC, CDS, and IC), and the bound on SE (EB1) along the vertical axis. DDC, CDS, IC, EC, SE, and  EB1 (= EB2) are represented using pink circles, green crosses, yellow pentagons, violet squares, brown stars, and black triangles respectively. All the quantities are dimensionless except ES, EC, and EB1, which have the dimension of ebits.}
\label{fig5} 
\end{figure}

EB1 and EB2 of phase flip channel are evaluated here only for one pair of $\{\boldsymbol{U}_{\textnormal{max}}$, $\boldsymbol{\rho}_{\textnormal{max}}\}$. Again we numerically observe that EB1 and EB2 are equal for each and every considered noise strength, $p$. We plot EB1 in Fig. \ref{fig5}, and realize that the graphical nature of EB1 (or EB2) is similar to that of SE for all noise strengths.

\subsection{Bit-phase flip channel}
Finally, we consider the bit-phase flip noise and determine the amount of saved entanglement, the measures of biasness, and the entanglement capacity. The bit-phase flip noise probabilistically swaps the eigenbasis of the $\sigma_z$ operator as well as adds a phase factor. 
Thus $\ket{0}\rightarrow -i \ket{1}$ and $\ket{1}\rightarrow i\ket{0}$. The transformation can be mathematically described as
\begin{equation*}
  \rho\rightarrow\Lambda_{BPF} (\rho) = \left(1-\frac{p}{2}\right) \rho + \frac{p}{2}  \sigma_y \rho \sigma_y,
  \end{equation*}
where $p$ is its noise strength.
The evaluated results are plotted in Fig. \ref{fig6} as functions of $p$. The nature of SE for the bit-phase flip channel, depicted in Fig. \ref{fig6}, can be explained in the similar way as in the preceding sub-sections, using biasness (DDC, CDS, and IC) in the lower noisy portion, i.e., $0 \leq p \leq \frac{1}{2}$, and using EC in the higher noisy region ($\frac{1}{2} \leq p \leq 1$).

\begin{figure}[h!]
\includegraphics[scale=1.0]{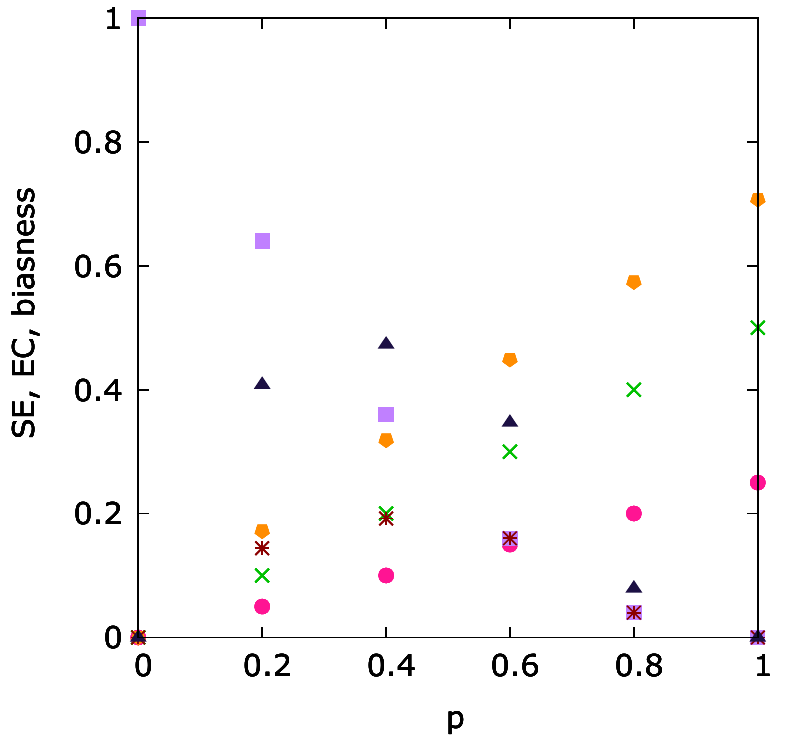}
\caption{Preserved entanglement from bit-phase flip noise, and biasness and entanglement capacity of the channel. All considerations are the same as in Fig. \ref{fig5}, except that the noise in context is bit-phase flip.} 
\label{fig6} 
\end{figure}

Furthermore, we also plot the bound EB1 in Fig. \ref{fig6}, which is again numerically equal to EB2 for all considered noise strengths. To obtain the bounds, we have considered only one set of $\{\boldsymbol{U}_{\textnormal{max}},\boldsymbol{\rho}_{\textnormal{max}}\}$. The nature of entanglement saved of the bit-phase flip channel can be described equivalently, by EB1 or EB2, as for the bit flip and phase flip channels. 

\section{Conclusion}
\label{sec7}
 Though entanglement is an essential resource in many quantum tasks including teleportation, dense coding, and entanglement-based cryptography, it is a fragile characteristic of shared quantum systems. Various unavoidable noise tend to reduce entanglement of shared quantum systems. Preservation of entanglement from such impact of noise is of significant practical interest. It was observed that if certain local unitaries are applied on the entangled state before the system's interaction with noise, the entanglement can be partially protected. The amount of entanglement that can be saved in this way depends on the nature of the noise, and as an extreme example, the depolarizing channel's effect can not be bypassed or diminished by utilizing local unitaries. In this work, we have tried to investigate the reason behind the partial protection provided by local unitaries. 

We explored the phenomenon through two physical characterstics of quantum channels, viz. biasness, which we argue as being able to explain the nature of saved entanglement when the strength of the applied noise is low, and entanglement capacity, which we argue as explaining the behaviour of saved entanglement for higher strengths of noise. We have also obtained two upper bounds on the saved entanglement, which we observed to represent the characteristics of the saved entanglement in the full range of noise strength.

\section*{Acknowledgment}
We acknowledge partial support from the Department of Science and Technology, Government of India through the QuEST grant (grant number DST/ICPS/QUST/Theme3/2019/120). The research of KS was supported in part by the INFOSYS scholarship.

\section*{appendix}
\label{sec8}
\begin{lemma}
Depolarizing channel is a covariant channel.
\end{lemma}

\begin{proof}
For the depolarizing channel expressed in Eq. \eqref{eq-dc},
\begin{align*}
    \mathcal{E}\left[\Lambda_{DC}\left(U \rho U^{\dagger}\right)\right] &= \mathcal{E}\left[\frac{p}{d}I_d + (1-p) U\rho U^{\dagger}\right],\\
    &= \mathcal{E}\left[\frac{p}{d} U U^{\dagger} + (1-p) U \rho U^{\dagger}\right],\\
    &= \mathcal{E}\left[U (\frac{p}{d} I_d + (1-p)\rho)U^{\dagger}\right],\\
    &= \mathcal{E}\left[U (\Lambda_{DC} (\rho)) U^{\dagger}\right],
\end{align*}
\end{proof}

\bibliography{bias_qm}

\end{document}